\documentclass[aps,superscriptaddress,twocolumn,showpacs,prl]{revtex4-2}

\usepackage{amsmath}
\usepackage{latexsym}
\usepackage{amssymb}
\usepackage{graphicx}
\usepackage[colorlinks=true, citecolor=blue, urlcolor=blue]{hyperref}
\usepackage{float}
\usepackage{amsfonts}
\usepackage{textcomp}
\usepackage{mathpazo}
\usepackage{blindtext}
\usepackage{changes}

\usepackage{bbm}

\usepackage{xcolor}
\definecolor{myurlcolor}{rgb}{0,0,0.4}
\definecolor{mycitecolor}{rgb}{0,0.5,0}
\definecolor{myrefcolor}{rgb}{0.5,0,0}
\usepackage{hyperref}
\hypersetup{colorlinks,
linkcolor=myrefcolor,
citecolor=mycitecolor,
urlcolor=myurlcolor}

\usepackage[draft]{fixme}
\usepackage{amsmath,bbm}
\usepackage{scrextend}
\usepackage{graphicx}
\usepackage{amsfonts}
\usepackage{amssymb}
\usepackage{amsmath,amssymb,amsthm,verbatim,graphicx,bbm}
\usepackage{mathrsfs}
\usepackage{mathtools}
\usepackage{color,xcolor,longtable}

\def\tr{\mbox{tr}}

\def\be{\begin{equation}}
\def\ee{\end{equation}}
\def\ben{\begin{eqnarray}}
\def\een{\end{eqnarray}}
\def\eea{\end{array}}
\def\bea{\begin{array}}

\newcommand{\Tr}[1]{\mathrm{Tr}#1}
\newcommand{\bei}{\begin{itemize}}
\newcommand{\eei}{\end{itemize}}
\newcommand{\ket}[1]{|#1\rangle}
\newcommand{\bra}[1]{\langle#1|}

\newcommand{\proj}[1]{\ket{#1}\!\bra{#1}}

\newcommand{\I}{\mathbbm{1}}

\newcommand{\dl}[1]{\left|\left|#1\right|\right|}

\renewcommand{\emph}[1]{\textbf{#1}}

\makeatletter
\newtheorem*{rep@theorem}{\rep@title}
\newcommand{\newreptheorem}[2]{%
\newenvironment{rep#1}[1]{%
 \def\rep@title{#2 \ref{##1}}%
 \begin{rep@theorem}}%
 {\end{rep@theorem}}}
\makeatother

\theoremstyle{plain}
\newtheorem{thm}{Theorem}
\newtheorem*{thm*}{Theorem}
\newreptheorem{thm}{Theorem}

\newtheorem{prop}[thm]{Proposition}
\newtheorem{cor}[thm]{Corollary}

\newtheorem*{ulem}{Lemma}

\theoremstyle{definition}

\theoremstyle{remark}

\usepackage[T1]{fontenc}

\usepackage{tikz}
\usetikzlibrary{arrows}

\tikzset{->, >=stealth', shorten >=1pt, auto, node distance=1cm, semithick, baseline=(current bounding box.center)}

\begin{document}


\title{Gap between quantum theory based on real and complex numbers is arbitrarily large}
\author{Shubhayan Sarkar}
\email{shubhayan.sarkar@ug.edu.pl}
\affiliation{Institute of Informatics, Faculty of Mathematics, Physics and Informatics,
University of Gdansk, Wita Stwosza 57, 80-308 Gdansk, Poland}

\author{David Trillo}
\affiliation{CUNEF Universidad, Calle Almansa 101, 28040 Madrid, Spain}

\author{Marc Olivier Renou}
\affiliation{Inria Paris-Saclay, Bâtiment Alan Turing, 1, rue Honoré d’Estienne d’Orves – 91120 Palaiseau 
}
\affiliation{CPHT, Ecole polytechnique, Institut Polytechnique de Paris, Route de Saclay – 91128 Palaiseau
}

\author{Remigiusz Augusiak}
\affiliation{Center for Quantum-Enabled Computing, Center for Theoretical Physics, Polish Academy of Sciences, Aleja Lotnik\'{o}w 32/46, 02-668 Warsaw, Poland}

\begin{abstract}	
Quantum information theory, the formalism for representing information contained in quantum systems, is based on complex Hilbert spaces. It was recently shown that in quantum networks involving three parties with nontrivial locality constraints, this formalism predicts correlations that cannot be explained by real quantum theory, a variant of quantum mechanics based on real Hilbert spaces.
In this work, we study a scenario with $N+1$ parties sharing quantum systems in a star network. We then construct a 
multipartite Bell inequality that exhibits a gap between quantum theory and its variant based on real numbers, which grows linearly with $N$, and is thus arbitrarily large in the asymptotic limit.
This implies, that, as the number of parties grows, Hilbert space formalism based on real numbers becomes exceedingly worse at describing complex networks of quantum systems. We also compute the tolerance of this gap to experimental errors.
\end{abstract}


\maketitle

\textit{Introduction.} Bell inequalities belong to the most paradigmatic ideas in quantum foundations, introduced in \cite{Bell} as a means to distinguish quantum mechanics from local realistic theories (like classical physics). Since then, extensive efforts have been dedicated to their study, ultimately leading to experimental verification that quantum systems do not exhibit local realistic behaviour \cite{Loophole1, Loophole2, Loophole3}. 
The typical scenario for a Bell experiment concerns two causally separated parties, Alice and Bob who share a distributed quantum state and perform local measurements on it. The result of the experiment is a collection of conditional probability distributions $\{p(a,b|x,y)\}$, referred to as correlations. A Bell inequality defined via a linear functional of these probabilities, provides a way to organize the correlations of the obtained experiment into a single number. One then compares the possible values this functional can achieve in different theories (typically, quantum mechanics and classical physics). If there is a difference, then one refers to a gap between the theories. Generally, the magnitude of this gap quantifies how well one theory approximates the other—for instance, how accurately local realistic correlations can approximate quantum correlations. Moreover, a larger gap enables greater confidence in experimentally distinguishing between the two theories \cite{AQ21}.

In this work, we adopt this perspective to compare two formulations of quantum theory: the standard one based on complex Hilbert spaces, referred to as complex quantum theory (CQT), and an alternative one based on real Hilbert spaces, known as real quantum theory (RQT). The need for using somewhat "unphysical" complex numbers in physical theories has always been the subject of intense debate. In the case of quantum theory, there have been attempts since its inception to prove that quantum-theoretical predictions can be described with real numbers \cite{real1, real2}. In the simplest case of a single system acting on a finite-dimensional Hilbert space, one can always reproduce any quantum correlations in a real Hilbert space of twice the dimension. Then, any correlations obtained in a Bell scenario, involving composite Hilbert spaces, can be simulated using RQT \cite{real2}. 

Only recently was it shown, using Bell violations within a quantum network-based setup, that real quantum theory cannot reproduce all quantum-mechanical predictions, confirming the need for complex numbers in quantum theory \cite{Marco}. 
This result immediately raises very fundamental questions: \textit{How far does quantum theory deviate from real quantum theory and is there a fundamental bound on the magnitude of the gap between both theories \cite{BB22,david}?} If such a fundamental bound exists, which is independent of the system size, RQT can still be a good approximation to CQT in the case of sufficiently large systems. 

Simultaneously, the gap size is relevant in the context of experimental verifications of the discrepancy between the theories in a loop-hole free setting. While the result of Ref. \cite{Marco} was already demonstrated in two experiments \cite{realexperi, RQTexp2}, these were not fully loop-hole free. Subsequently, the locality loophole was closed in \cite{localityexperi}, but not the detection loophole. The larger the gap between the maximum values of the corresponding Bell-type inequality obtained in CQT and RQT, the lower the detectors efficiency needed to experimentally verify the difference between theories. While the search for larger gaps between RQT and CQT and better experimental conditions has recently become an active research area \cite{Batle_2025,Yao2024,Lancaster_2025, sarkaropead}, 
the above questions remain unresolved.

In this work, we address these questions and show that no such fundamental bound exists; in fact, quantum theory can depart arbitrarily far from its real counterpart in terms of Bell violations. We achieve this goal by generalizing the scenario introduced in \cite{Marco} to the well-known star network. 
%
We first consider the operational condition that the states and measurements in the network must satisfy for both RQT and CQT. Based on it, we then define a functional that allows us to witness a gap between these two formulations of quantum theory that increases with the number of parties. To ensure that the gap is genuine, we normalize the witness so that its maximal and minimal algebraic values remain fixed for any number of parties. 
We call this witness a "conditional" Bell inequality because the functional is optimized not over all possible states and measurements but only those satisfying the operational condition. Additionally, we compute the CQT--RQT gap in the case when one cannot exactly satisfy the operational condition but can only observe it up to some small error. Importantly, our results are all analytical. 

In fact, a similar question as to whether there exists an unbounded gap between quantum and local realistic correlations has been naturally asked after Bell's seminal work \cite{Bell}. Despite one might naturally assume that increasing the size of a quantum system typically leads to larger Bell violations, this intuition is not true for most instances \cite{Flavio, SATWAP, CGLMP, Augusiak_2019} [see \cite{Brunner_2014} for more examples]. Yet, despite many efforts, we could identify only three works in the literature that have rigorously demonstrated 
an unbounded growth of the gap (cf. Refs. \cite{PhysRevLett.65.1838,v008a027, Khot2005}). Consequently, finding a scenario and then constructing functionals that exhibit an unbounded gap between two physical theories is an extremely challenging task. Moreover, in Bell scenarios, quantum advantage originates from entanglement and incompatible measurements, enabling correlations unattainable by any local hidden-variable theory. In our case, the gap between CQT and RQT stems from the tensor-product composition of two subsystems and the complex, non-commuting observables in CQT, allowing us to show that imaginarity is a resource that can provide unbounded advantage.

\textit{Preliminaries.} The quantum network considered here consists of $N$ external parties, denoted $A_i$ $(i=1,\ldots, N)$, and a central party $E$ called Eve. All the parties are spatially separated. Furthermore, the network consists of $N$ independent sources that generate bipartite quantum states, and send one to the external party and the other to Eve as depicted in Fig. \ref{fig1}. On their shares of the joint state, each external party $A_i$ performs one of three available measurements, each with two outcomes; the measurement choices and outcomes of party $A_i$ are denoted $x_i=0,1,2$ and $a_i=0,1$, respectively. Then, the central party Eve performs a single measurement with $2^N$ outcomes, where her outcomes are denoted by $l=0,\ldots,2^N-1$. 
Importantly, the parties cannot communicate classically during the experiment. From the above scenario, one obtains the probability distribution $\vec{p}=\{p(\mathbf{a}l|\mathbf{x})\}$ where $p(\mathbf{a}l|\mathbf{x})$ is the probability of obtaining outcomes $a_1\ldots a_N=:\mathbf{a}$ by the external parties after they perform the measurements labeled by $x_1\ldots x_N=:\mathbf{x}$ and $l$ by Eve, respectively.

Within quantum theory based on complex Hilbert space formalism these probabilities are expressed as
\begin{equation}\label{probs}
p(\mathbf{a}l|\mathbf{x})=\Tr\left(\rho_{\mathbf{A}E}\bigotimes_{i=1}^{N} M_{a_i|x_i} \otimes R_{l}\right),
\end{equation}
where $\{M_{a_i|x_i}\}$, for every $i$, and $\{R_{l}\}$ are the POVM elements corresponding to outcome $a_i$ of the measurement $x_i$ performed by party $A_i$; and outcome $l$ of the measurement performed by $E$, respectively. Then, the joint state reads $\rho_{\mathbf{A}E}=\otimes_i\rho_{A_iE_i}$ with $\rho_{A_iE_i}$ denoting the state generated by the source $i$ which is distributed between $A_i$ and $E$. The measurement elements $M_{a_i|x_i}$ and $R_{l}$ are positive semi-definite and sum up to the identity on their respective Hilbert spaces. 
\begin{figure}[t]
    \centering
      \begin{tikzpicture}
    \node[draw, circle] (x) {$A_1$};
    \node[draw, circle] (y) [right of = x, node distance = 3cm] {$A_2$};
    \node[] (...) [right of = y, node distance = 1.5cm] {...};
    \node[draw, circle] (z) [right of = y, node distance = 3cm] {$A_n$};
    \node[draw] (a) [below of = x, node distance = 1.5cm] {$\rho_{A_1E_1}$};
    \node[draw] (b) [below of = y, node distance = 1.5cm] {$\rho_{A_2E_2}$};
    \node[] (...2) [right of = b, node distance = 1.5cm]{...};
    \node[draw] (c) [below of = z, node distance = 1.5cm] {$\rho_{A_N E_N}$};
    \node[draw, circle] (C) [below of = b, node distance = 1.5cm] {$E$};
    \node[draw, circle] (X) [above of = x, node distance = 1.5cm] {$x_1$};
    \node[draw, circle] (Y) [above of = y, node distance = 1.5cm] {$x_2$};
    \node[] (...3) [right of = Y, node distance = 1.5cm] {...};
    \node[draw, circle] (Z) [above of = z, node distance = 1.5cm] {$x_n$};
    
    \draw[->] (a) to node {} (x);
    \draw[->] (c) to node {} (z);
    \draw[->] (a) to node {} (C);
    \draw[->] (b) to node {} (C);
    \draw[->] (c) to node {} (C);
    \draw[->] (b) to node {} (y);
    \draw[->] (X) to node {} (x);
    \draw[->] (Y) to node {} (y);
    \draw[->] (Z) to node {} (z);
    \end{tikzpicture}
    \caption{\textbf{Depiction of the quantum network scenario.} 
    It consists of $N+1$ parties, namely, $A_i$ $(i=1,\ldots,N)$, and $E$, and $N$ independent sources distributing bipartite quantum states $\rho_{A_iE_i}$ among the parties as shown in the figure. The central party $E$ shares quantum states with each one of the other external parties $A_i$. While each $A_i$ has three inputs (shown as classical inputs to each of the parties) and two outcomes, $E$ has a single input with $2^N$ outputs.}
    \label{fig1}
\end{figure}
In what follows, it is beneficial to express the Bell inequalities in terms of the expectation values defined as $\langle A_{1,x_1}\ldots A_{N,x_N} R_{l} \rangle
     = \sum_{a_1\ldots a_N=0}^{1} (-1)^{\sum_{i=1}^N a_i} p(\mathbf{a}l|\mathbf{x}).$
     These, by virtue of Eq. \eqref{probs}, can also be expressed as $\langle A_{1,x_1}\ldots A_{N,x_N} R_{l} \rangle = \Tr[(\bigotimes_{i=1}^{N} A_{i,x_i}) \otimes R_{l}\rho_{AE}]$, where $A_{i,x_i}$ are Hermitian operators defined via the measurement elements as $A_{i,x_i}=M_{0|x_i}-M_{1|x_i}$ for every $x_i$ and $i$. Generally, the corresponding observables satisfy $A_{i,x_i}^2\leq\I$ with the equality holding iff the measurement is projective. For our convenience, Eve's measurement is represented in the above expectation values in terms of the measurement elements $R_l$.


Real quantum theory uses the same formalism as defined above but based on real Hilbert spaces. More precisely, all states and measurements in it are constrained to be real in some fixed basis; that is, $\rho_{A_iE_i}=\rho_{A_iE_i}^*$ and $M_{a_i|x_i}=M_{a_i|x_i}^*,\ R_l=R_l^*$ where $*$ denotes the complex conjugation with respect to the fixed basis. Hence, the corresponding observables $A_{i,x_i}$ are also real. Importantly, the tensor product rule stays unchanged: independent systems are represented by product states. Note that for convenience, we use standard product notation in place of the tensor product symbol $\otimes$ where appropriate. The Hilbert space on which a local operator acts can be inferred from context. 

\textit{Results.} Finding the gap between CQT and RQT using the standard approach of evaluating the maximal values of a Bell-like functional over CQT and RQT is challenging. In particular, no general analytical way to compute the maximal value of the considered functional achievable in RQT. As we want to find those values in the star network scenario (cf. Fig. \ref{fig1}) for any $N$, numerical techniques are insufficient, and an analytical approach is required to find the CQT–RQT gap. 

The core idea of our result makes use of self-testing \cite{Mayers2004}, which is the most complete form of device-independent certification enabling almost full characterization of the quantum state and measurements, up to local unitaries and junk states, giving rise to the observed non-classicality. For this purpose, we follow an approach recently introduced in \cite{sarkar2023UPB, sarkar2023universal}.

Here, we will use two Bell functionals: $\mathcal{I}_l$ defined in \eqref{BE1Nm} which is a part of the functional introduced in \cite{sarkar2023universal}, and $\mathcal{J}_N$ defined in \eqref{BE2m}, which is a part of Mermin's functional \cite{Mermin}, to construct a "conditional" Bell functional,
\begin{eqnarray}
    \max_{\forall l,\ \mathcal{I}_l=2(N-1),\ \overline{P}(l)=1/2^N}\mathcal{J}_N,
\end{eqnarray}
which will witness the desired gap. The idea behind our construction is that the maximal quantum value of $\mathcal{I}_l$ self-tests Eve's measurements and the states distributed by the sources. 
Then, computing the maximal values of $\mathcal{J}_N$ under the operational condition that $\mathcal{I}_l=2(N-1)$, in both RQT and CQT, will allow us to witness the gap between both theories. In an actual experiment, one would collect all the statistics necessary to compute both functionals. If both achieve their maximal quantum values, this demonstrates that RQT cannot fully describe the experimental setup.

Let us now introduce the Bell inequalities mentioned above to be $\mathcal{I}_l=\langle\hat{\mathcal{I}}_l\rangle\leq\beta_C$ with $\hat{\mathcal{I}}_{l}$ defined as
\begin{equation}\label{BE1Nm}
\hat{\mathcal{I}}_{l}=(-1)^{l_1}\left[ (N-1)\tilde{A}_{1,1}\bigotimes_{i=2}^N A_{i,1} +\sum_{i=2}^N(-1)^{l_i}\tilde{A}_{1,0}\otimes A_{i,0}\right],
\end{equation}
where $l\equiv l_1\ldots l_N$ such that $l_1,l_2,\ldots,l_N=0,1$ is the binary representation of the outcome $l$, and $ \tilde{A}_{1,i}=(A_{1,0}-(-1)^i A_{1,1})
/\sqrt{2}$. 
The maximal classical values of $\mathcal{I}_l$ are directly found to be $\beta_C=\sqrt{2}(N+1)$.

Then, the maximal quantum values of $\mathcal{I}_l$, also referred to as the quantum or Tsirelson's bounds, equal $2(N-1)$ and are achieved by the following observables 
\begin{eqnarray}\label{GHZObsm}
A_{1,0}'&=&\frac{X+Z}{\sqrt{2}},\qquad A_{1,1}'= \frac{X-Z}{\sqrt{2}} \nonumber\\
A_{i,0}'&=&Z,\qquad A_{i,1}'=X\qquad i=2,\ldots,N
\end{eqnarray}
and the Greenberger-Horne-Zeilinger (GHZ) states

\begin{equation}\label{GHZvecsm}
\ket{\phi_l}=\frac{1}{\sqrt{2}}(\ket{l_1\ldots l_N}+(-1)^{l_{1}}|\overline{l}_1\ldots\overline{l}_N\rangle),
\end{equation}
where $l_1,l_2,\ldots,l_N=0,1$ and $\overline{l}_i=1-l_i$. 
Notice that the measurements \eqref{GHZObsm} and the states \eqref{GHZvecsm} are real, and hence one can attain the quantum bounds of $\mathcal{I}_l$ within RQT.

Consider again the scenario depicted in Fig. \ref{fig1}. If the Bell inequality $\mathcal{I}_l$ \eqref{BE1Nm} is maximally violated in this scenario when Eve obtains the outcome $l$, occurring with probability $\overline{P}(l)=1/2^N$, then one can self-test the state and measurements in CQT to be the ideal ones (\eqref{GHZvecsm} and \eqref{GHZObsm} respectively) up to local unitaries and junk states. A rigorous statement of this can be found in the Supplemental Material. In this self-testing statement, it is always possible to restrict the unitaries and junk states obtained to be real, if we work under the assumption that RQT describes the experiment.

Now, we introduce the second Bell functional used in our work to be 
\begin{eqnarray}\label{BE2m}
\mathcal{J}_{N}=-\frac{2}{N(N-1)}\left[\sum_{\substack{j,k=2\\j<k}}^N\left\langle \tilde{A}_{1,1}A_{j,2}A_{k,2}\prod_{\substack{i=2\\i\ne j,k}}^{N}  A_{i,1}\right\rangle\right.\nonumber\\
\left.+\sum_{j=2}^N\left\langle A_{1,2}A_{j,2}\prod_{\substack{i=2\\i\ne j}}^{N} A_{i,1} \right\rangle\right].\qquad
\end{eqnarray}
Recall that any linear functional can be rescaled by multiplying it by a scalar or adding a constant term. As we compare the values of the above functional for any $N$, it is important to ensure that the gap we obtain is genuine and not simply a consequence of scaling the functional for higher $N$. The functional is thus normalized so that its values lie within the range $-1\leq\mathcal{J}_{N}\leq1$ for any $N$. 

Let us now consider the scenario in Fig. \ref{fig1} and suppose that all the external parties observe the maximal violation of the Bell expressions $\mathcal{I}_l$ [cf. Eq. \eqref{BE1Nm}] to certify the states generated by the sources and the measurements $A_{i,0},A_{i,1}$, as proven in Theorem 1 of the Supplementary Materials. Furthermore, when Eve obtains the outcome $l=l_1\ldots l_N$ with $l_i=0$ for all $i$, then they also evaluate the functional $\mathcal{J}_N$ \eqref{BE2m}. 

Within standard quantum theory, one can simultaneously achieve the maximal values of all $\mathcal{I}_l$ and the expression $\mathcal{J}_{N}$ with a value of $\beta_{\mathrm{CQT}}=1$. This is achieved by the GHZ state and observables 
in Eq. (\ref{GHZObsm}) and 
$A_{i,2}=Y$ for all $i$.
%
This ensures that each of the $N(N-1)/2$ terms in the brackets of Eq. \eqref{BE2m} is $-1$. 


Importantly, if, on the other hand, one restricts to RQT, then the situation is as follows:
\begin{thm}\label{thm1}
   In the scenario depicted in Fig. \ref{fig1}, we restrict to RQT. Now assume that the external parties obtain the maximal violation of the Bell inequalities $\mathcal{I}_{l}$ for all Eve's outcome $l$. Then, with Eve obtaining outcome $l=0\ldots0$, the maximal value of $\mathcal{J}_N$ is $\beta_{\mathrm{RQT}}\leq 1/(N-1)$ for any $N$.
\end{thm}
Consequently, given the maximal violation of the Bell inequalities $\mathcal{I}_l$ \eqref{BE1Nm}, one can observe that the ratio of maximal value of $\mathcal{J}_N$ with CQT to RQT is $\beta_{\mathrm{CQT}}/\beta_{\mathrm{RQT}}\geq N-1$, which increases linearly with $N$. Note that Theorem \ref{thm1} does not allow us to conclude that there is a gap in the original network considered in \cite{Marco} with three binary measurements for the external parties and a single eight-outcome measurement, which is the bilocality network or equivalently the star network with $N=2$.

The proof of Theorem \ref{thm1} is deferred to the Supplemental Material. Its key step to find the value of $\mathcal{J}_N$ is to use the maximal violation of $\mathcal{I}_l$ to self-test the measurement of Eve $\{R_l\}$ and other parties observables $A_{i,j}$ for $j=0,1$ to be the ideal ones as stated in Eqs. \eqref{GHZvecsm} and \eqref{GHZObsm} respectively. Along with it, the states generated by the sources $\rho_{A_iE_i}$ are certified to be maximally entangled up to some local unitaries and junk states [see Theorem 1 of the supplementary material]. Another crucial aspect of the self-testing proof is that the obtained junk states are product states. 
Then, imposing the constraint that all states and measurements are real with respect to some basis, we find certain relations must hold for the expectation values in Eq. \eqref{BE2m} with the self-tested states and the observables $A_{i,1}$, which then allows us to find an upper bound of $\mathcal{J}_N$ [see Theorem 3 of the supplementary material]. 
In Theorem 1, for readability, we focused on  the case where we condition on Eve obtaining outcome. However, in an experiment, it would be more practical to use a witness $\sum_{l} \overline{P}(l) \mathcal{J}_N^{l}$, that includes all Eve's outcomes, avoiding any loss of statistics. This can be done, just as in \cite{Marco}, without changing the gap, by adding a suitable sign correction $(-1)^{f(l_1,\ldots,l_N)}$, where $f(l_1,\ldots,l_N)$ is some function of outcome indices, to each expectation value in $\mathcal{J}_N$. This is done explicitly in Appendix C.

Let us now analyse the above setup in the noisy scenario. We are concerned here with determining what happens when one does not exactly achieve the quantum value of $\mathcal{I}_l$ \eqref{BE1Nm}, but only a value $\varepsilon_N-$close to it, that is, $\langle\psi_l|\hat{\mathcal{I}}_{l}|\psi_l\rangle\geq 2(N-1)-\varepsilon_N$ for all $l$. We also allow for noise in the probability of the outcomes of the measurement of Eve, that is, we consider the situation that $|\overline{P}(l)-1/2^N|\leq\varepsilon_N$  for all $l$.

We first prove a noisy version of the self-testing statement of the maximal value of $\mathcal{I}_l$ in the Supplemental Material that maintains all the key ingredients that we need to find a gap witnessed by $\mathcal{J}_N$. We first certify that the post-measured states of Eve which are shared by all the external parties are close to the ideal state (up to local unitaries and junk part) up to an error of $N^2\sqrt{\varepsilon_N}$. Along with it we also certify that the observables of all parties acting on the local state are close to the ideal ones (up to local unitaries and junk part) up to an error of $N^2\sqrt{\varepsilon_N}$. Finally, we show that the junk part of the post-measured states are close to some product state up to an error of the order of $N^6\sqrt{\varepsilon_N}$.   

Consequently, using the certified states and measurements, we find the following upper bound to $\beta_{\mathrm{RQT}}$ in the noisy scenario:
\begin{thm}
    Consider the scenario depicted in Fig.1 and suppose that the Bell functionals $\langle\hat{\mathcal{I}}_{l}\rangle$ attain a value $\varepsilon_N$-close to the quantum bound, and the probability of the outcomes of the measurement of Eve satisfy $|\overline{P}(l)-1/2^N| \leq\varepsilon_N$ for all $l$.  Then, the maximal value of $\mathcal{J}_N$ for RQT is upper bounded as
\begin{eqnarray}
\beta_{\mathrm{RQT}}\leq \frac{1}{N-1}+ O(N^6\sqrt{\varepsilon_N}).
\end{eqnarray}
\end{thm}
The proof is stated in the Supplemental Material, where we provide an exact expression for the upper bound of $\beta_{\mathrm{RQT}}$. Even though we obtain that the error has to be suppressed in the number of parties to observe a gap, let us remark here that the actual value for any particular $N$ is expected to be much smaller than our rough estimates. Indeed, it is possible to find better numerical bounds for the case $N=3$ with semidefinite-programming techniques. In Fig. \ref{plot}, we plot the the noise tolerance of the RQT value to noise. The details of the method is provided in Appendix C (and the code is freely available in the Supplemental Material and in \cite{code}). 
\begin{figure}
    \centering
    \includegraphics[width=\linewidth]{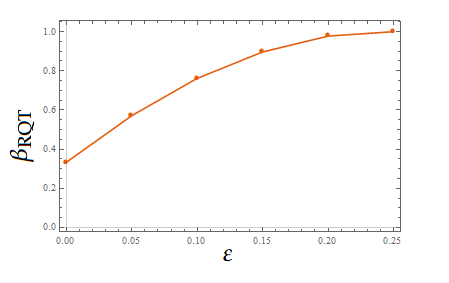}
    \caption{The maximal value of the functional \eqref{BE2m} for $N=3$, acheivable in RQT. Here $\varepsilon$ denotes the deviation from the optimal value of the functional \eqref{BE1Nm}. The solid points indicate the obtained numerical values.}
    \label{plot}
\end{figure}


{\it{Discussion.}} We have demonstrated that real quantum theory becomes increasingly inaccurate at reproducing quantum correlations in the star network as the number of parties grows. This provides a negative answer to a key foundational question in the field as to whether there is a fundamental bound on the size of the gap between CQT and RQT. We have witnessed this increasing separation between the two sets of correlations with a conditional Bell inequality, and we have shown that this separation is robust to small errors in Bell violations. We have furthermore performed a detailed numerical analysis in the case $N=3$ to show that the gap is actually much bigger than what can be proven with analytical techniques. Such an unbounded gap is known to exist between quantum correlations local correlations which are those that admit local hidden variable models \cite{Mermin,KhotVishnoi2005,Perez_Garc_a_2008,PhysRevA.92.052313} even with a fixed number of parties. It is surprising that RQT and CQT also exhibit an unbounded gap, given that they are, as generalized probabilistic theories, way more similar \cite{ying2025}. We leave for future work the question of whether this unbounded gap can be realized within a fixed number of parties.

Until now, research on the real–complex separation in quantum theory has mainly aimed at increasing the magnitude of the gap magnitude while simultaneously simplifying the scenarios considered from the experimental point of view. Our work suggests that larger networks can be better suited to close the still-open detection loophole of current experiments, as they allow the gap size to become arbitrarily large. To this end, a more detailed study of the experimental feasibility of our result would be necessary.

Let us also remark that enforcing the independence of sources in an experimental setting is indeed challenging. This is primarily because past events or common causes may induce correlations between the sources that cannot be fully accounted for. However, this loophole is similar to the free-will loophole in the standard Bell setting. In a recent work, \cite{PRL121211} the gap was shown to persist when only some partial independence among sources is present, however, one can not completely drop the assumption on independence. Nevertheless, we maintain that genuinely independent sources do exist in nature. Without them, the very foundations of physics as we understand them, would be undermined. It would also be interesting to identify a standard linear Bell inequality capable of witnessing such a large separation. 



\begin{center}
    \textbf{Data Availability}
\end{center}

No data set was generated or analyzed in the current study.

\begin{center}
    \textbf{Code Availability}
\end{center}

The code written is freely available in \cite{code}.

\begin{center}
    \textbf{Acknowledgements}
\end{center}
 RA and SS acknowledge the QuantERA II Programme (VERIqTAS project) that has received funding from the European Union’s Horizon 2020 research and innovation programme under Grant Agreement No 101017733. SS also acknowledges 
the National Science Centre, Poland, grant Opus 25,
UMO-2023/49/B/ST2/02468. MOR acknowledges funding by the ANR for the JCJC grant LINKS (ANR-23-CE47-0003) and INRIA and CIEDS Action Exploratoire project DEPARTURE.

\begin{center}
    \textbf{Author contributions}
\end{center}
S.S., D.T., M.O.R, and R.A. conceived the idea, designed the proofs and prepared the manuscript.

\begin{center}
    \textbf{Competing interests}
\end{center}
The authors declare no competing interests.

\appendix
\onecolumngrid

\section{Appendix A: Scalable gap in the noiseless scenario}

Consider the scenario depicted in Fig. 1 of the manuscript and assume that the external parties 
obtain the maximal violation of $\mathcal{I}_{l}=\langle\mathcal{I}_{l}\rangle$ for all Eve's outcome $l$ where $l\equiv l_1\ldots l_N$ and the corresponding Bell operators are given by
\begin{equation}\label{BE1N}
\hat{\mathcal{I}}_{l}=(-1)^{l_1}\left[ (N-1)\tilde{A}_{1,1}\bigotimes_{i=2}^N A_{i,1} +\sum_{i=2}^N(-1)^{l_i}\tilde{A}_{1,0}\otimes A_{i,0}\right],
\end{equation}
$ \tilde{A}_{1,0}=(A_{1,0}-A_{1,1})
/\sqrt{2},\qquad\tilde{A}_{1,1}=(A_{1,0}+A_{1,1})/\sqrt{2}.$
This allows us to utilize the certified states generated by the sources and the measurements $A_{i,1},A_{i,2}$ as stated below. For reference below,
\begin{eqnarray}\label{GHZObs}
A_{1,0}'&=&\frac{X+Z}{\sqrt{2}},\qquad A_{1,1}'= \frac{X-Z}{\sqrt{2}} \nonumber\\
A_{i,0}'&=&Z,\qquad A_{i,1}'=X\qquad i=2,\ldots,N
\end{eqnarray}
and,

\begin{equation}\label{GHZvecs}
\ket{\phi_l}=\frac{1}{\sqrt{2}}(\ket{l_1\ldots l_N}+(-1)^{l_{1}}|\overline{l}_1\ldots\overline{l}_N\rangle).
\end{equation}

\setcounter{thm}{0}
\begin{thm}\label{theorem1}
Assume that the Bell inequalities \eqref{BE1N} for any $l$ are maximally violated and each outcome of the central party occurs with probability $\overline{P}(l)=1/2^N$ where $N$ denotes the number of external parties. The sources $P_i$ prepare the states $\ket{\psi_{A_iE_i}}\in \mathcal{H}_{A_i}\otimes\mathcal{H}_{E_i}$ for $i=1,2,\ldots,N$, the measurement with the central party is given by $\{R_{l}\}$ for $l=l_1l_2\ldots l_N$ such that $l_i=0,1$ which acts on $\bigotimes_{i=1}^N\mathcal{H}_{E_i}$, the local observables for each of the party is given by $A_{i,j}$ for $i=1,2,\ldots,N,\ j=0,1$ which act on $\mathcal{H}_{A_i}$. Then,
\begin{enumerate}
    \item The Hilbert spaces of all the parties decompose as $\mathcal{H}_{A_i}=\mathcal{H}_{A_i'}\otimes\mathcal{H}_{A_i''}$ and $\mathcal{H}_{E_i}=\mathcal{H}_{E_i'}\otimes\mathcal{H}_{E_i''}$, where $\mathcal{H}_{A_i'}$ and 
    $\mathcal{H}_{E_i'}$ are qubit Hilbert spaces, whereas $\mathcal{H}_{A_i''}$ and $\mathcal{H}_{E_i''}$
    are some finite-dimensional but unknown auxiliary Hilbert spaces.
    \item There exist local unitary transformations $U_{A_i}:\mathcal{H}_{A_i}\rightarrow(\mathbb{C}^2)_{A'}\otimes\mathcal{H}_{A''_i}$ and $U_{E_i}:\mathcal{H}_{E_i}\rightarrow(\mathbb{C}^2)_{A'}\otimes \mathcal{H}_{E''_i}$ 
    such that the states are given by
\begin{eqnarray}\label{statest1}
U_{A_i}\otimes U_{E_i}\ket{\psi_{A_iE_i}}=\ket{\phi^+_{A_i'E'_i}}\otimes\ket{\xi_{A_i''E''_i}}
\end{eqnarray}
where $\ket{\phi^+}=1/\sqrt{2}(\ket{00}+\ket{11})$.
\item The measurement of the central party is
\begin{eqnarray}\label{A10}
 U_ER_lU_E^{\dagger} =\proj{\phi_l}_{E'}\otimes\I_{E''}\qquad \forall l,
\end{eqnarray}
where $U_E=U_{E_1}\otimes\ldots\otimes U_{E_N}$ and $\ket{\phi_l}$ are given in Eq. \eqref{GHZvecs}. The measurements of all the other parties are given by
\begin{eqnarray}\label{mea1}
U_{A_i}A_{i,j}\,U_{A_i}^{\dagger}&=&A_{i,j}'\otimes\I_{A_i''}\quad \forall i,j
\end{eqnarray}
where $A_{i,j}'$ are given in \eqref{GHZObs}.
\end{enumerate}
\end{thm}
\begin{proof}
    Here, we outline the main steps of the proof when considering CQT with the details deferred to Ref. \cite{sarkar2023universal}. We begin by considering the sum-of-squares (SOS) decomposition of the Bell operator $\hat{\mathcal{I}}_l$. When Bell functional attains the maximal quantum value, the SOS decomposition allows us to prove that measurements of all the parties anti-commute, that is, $\{A_{i,0},A_{i,1}\}=0$ for all $i$. Then choosing appropriate unitaries, we can straightaway obtain \eqref{mea1}. Using them to solve the relations obtained from the SOS decomposition, we obtain that the state that attains the maximal quantum value of $\mathcal{I}_l$ for all $l$ is given by Eq. \eqref{GHZvecs} up to local unitaries and junk state. Recall that, these are post-measurement states of Eve's particular outcomes. Thus, utilising these certified states and the fact that Eve's outcomes are equiprobable, that is, $\overline{P}(l)=1/2^N$, we self-test first the states generated by the source to be \eqref{statest1} and then Eve's measurements \eqref{A10}.

    The same theorem and proof works in the case of RQT as the ideal states and measurements are real. The only difference lies in the unitaries $U_{A_i}$ that need to be real, that is, $U_{A_i}=U_{A_i}^*$ along with real junk states $\ket{\xi_{A_i''E''_i}}=\ket{\xi_{A_i''E''_i}^*}$ for all $i$.
\end{proof}

Now, when Eve's outcome is $l\equiv l_1\ldots l_N$ where $l_i=0$ for all $i$, all the external parties evaluate the following functional inspired from \cite{Mermin}
\begin{eqnarray}\label{BE2}
\mathcal{J}_{N}=-\frac{2}{N(N-1)}\left[\sum_{\substack{j_1,j_2=2\\j_1<j_2}}^N\left\langle \tilde{A}_{1,1}A_{j_1,2}A_{j_2,2}\prod_{\substack{i=2\\i\ne j_1,j_2}}^{N}  A_{i,1}\right\rangle+\sum_{j_1=2}^N\left\langle A_{1,2}A_{j_1,2}\prod_{\substack{i=2\\i\ne j_1}}^{N} A_{i,1} \right\rangle\right],
\end{eqnarray}
with $\tilde{A}_{1,1}=\frac{A_{1,0}+A_{1,1}}{\sqrt{2}}$.
Notice that the above functional \eqref{BE2} consists of $\frac{N(N-1)}{2}$ terms. As each of these terms is upper bounded by $1$, we obtain that $-1\leq\mathcal{J}_{N}\leq1$ for any $N$. 

Let us now restrict to CQT and find the maximal value of $\mathcal{J}_N$.

\setcounter{fakt}{0}

\begin{thm}
   In the scenario depicted in Fig. 1 of the manuscript, we restrict to CQT. Now assume that the external parties obtain the maximal violation of the Bell inequalities $\mathcal{I}_{l}$ for all Eve's outcome $l$. Then, with Eve obtaining outcome $l=0\ldots0$, the maximal value of $\mathcal{J}_N=1$ for any $N$.
\end{thm}
\begin{proof}
    From Theorem \ref{theorem1}, the states generated by the sources are certified as
    \begin{eqnarray}
U_{A_i}\otimes U_{E_i}\ket{\psi_{A_iE_i}}=\ket{\phi^+_{A_i'E'_i}}\otimes\ket{\xi_{A_i''E''_i}}.
\end{eqnarray}
and the measurements of all the external parties are certified as
\begin{eqnarray}
U_{A_i}A_{i,j}\,U_{A_i}^{\dagger}&=&A_{i,j}'\otimes\I_{A_i''}
\end{eqnarray}
for all $i$ and $j=0,1$ where $A_{i,j}'$ are given in Eq. 5 of the manuscript. Now, choosing
\begin{eqnarray}
    U_{A_i}A_{i,2}\,U_{A_i}^{\dagger}=Y\otimes\I_{A_i''}
\end{eqnarray}
one can straightforwardly obtain that each term inside the large bracket in Eq. \eqref{BE2} is $-1$ and thus $\mathcal{J}_N=1$ for any $N$.
\end{proof}

Let us now restrict to RQT to find the maximal value of $\mathcal{J}_N$.
\begin{thm}
   In the scenario depicted in Fig. 1 of the manuscript, we restrict to RQT. Now assume that the external parties obtain the maximal violation of the Bell inequalities $\mathcal{I}_{l}$ for all Eve's outcome $l$. Then, with Eve obtaining outcome $l=0\ldots0$, the maximal value of $\mathcal{J}_N\leq 1/(N-1)$ for any $N$.
\end{thm}
\begin{proof}
From Theorem \ref{theorem1}, the states generated by the sources are certified as
\eqref{statest1} with $\ket{\xi_{A_i''E''_i}}=\ket{\xi^*_{A_i''E''_i}}$ and the measurements $A_{i,0}, A_{i,1}$ are certified as in \eqref{mea1}. Consequently, the local Hilbert spaces of the external parties are even-dimensional. Now, any operator $U_{A_i}A_{i,2}U_{A_i}^{\dagger}$ acting on an even-dimensional Hilbert space can be expressed as 
\begin{eqnarray}
U_{A_i}A_{i,2}U_{A_i}^{\dagger}=\sum_j\sigma_j\otimes r_{i,j}
\end{eqnarray}
where $\sigma_{0}=\I, \sigma_{1}=Z,\sigma_2=X, \sigma_3=Y$ and $r_{i,j}$ are operators acting on $\mathcal{H}_{A_i''}$. As $A_{i,2}$ is Hermitian, we obtain that $r_{i,j}=r_{i,j}^{\dagger}$. As we restrict to RQT, we also obtain from the condition $A_{i,2}=A_{i,2}^*$ that
\begin{eqnarray}\label{consreal2}
r_{i,j}=r_{i,j}^*\quad (j=0,1,2),\quad \text{and}\quad r_{i,3}=-r_{i,3}^*.
\end{eqnarray}
Furthermore, with any state $\rho$ belonging to real Hilbert space, that is, $\rho=\rho^*$ we have that
\begin{eqnarray}\label{use1}
    \Tr(r_{i,3}\rho)=\Tr([r_{i,3}^*\rho^*]^*)= -\Tr([r_{i,3}\rho]^*)=-[\Tr(r_{i,3}\rho)]^*.
\end{eqnarray}
Similarly, using the fact that $r_{i,j}=r_{i,j}^{\dagger}$, we have that
\begin{eqnarray}\label{use2}
    \Tr(\rho\ r_{i,3})= \Tr(r_{i,3}\rho)=\Tr([\rho^\dagger r_{i,3}^\dagger]^\dagger)= \Tr([\rho\ r_{i,3}]^\dagger)=[\Tr(\rho\ r_{i,3})]^*.
\end{eqnarray}
Consequently, for any state $\rho=\rho^*$ we have from \eqref{use1} and \eqref{use2} that $\Tr(r_{i,3}\rho)=0$. Furthermore, we also have that $\Tr(r_{i,j}\rho)=[\Tr(r_{i,j}\rho)]^*$ for $j=0,1,2$. This condition will be particularly useful later on in the proof.

Let us now notice that the post-measured state of the external parties when Eve obtains the outcome $l=0\ldots0$ can be expressed using Theorem \ref{theorem1} as
\begin{eqnarray}\label{post-mea1}
    \tilde{\rho}_{0\ldots 0}= \left(\prod_{i=1}^NU_{A_i}\right)\rho_{0\ldots 0} \otimes \left(\prod_{i=1}^NU_{A_i}^{\dagger}\right)=\proj{\phi_{0\ldots0}}\prod_{i}^N\rho_i
\end{eqnarray}
where $\phi_{0\ldots0}$ is the $N-$qubit GHZ state \eqref{GHZvecs} and $\rho_i=\Tr_{E_i''}(\xi_{A_i''E_i''})$. Consequently, we have in RQT that $\rho_i=\rho^{*}_i$. Now, let us evaluate $\mathcal{J}_N$ \eqref{BE2} using the above decomposition of $A_{i,2}$ \eqref{consreal2} and the post-measured state \eqref{post-mea1}. Considering the first term in \eqref{BE2}, we have that
\begin{eqnarray}
  \left\langle \tilde{A}_{1,1}A_{j_1,2}A_{j_2,2}\prod_{\substack{i=2\\i\ne j_1,j_2}}^{N} A_{i,1}\right\rangle= \Tr\left( U_{A_{j_1}}A_{j_1,2}U_{A_{j_1}}^{\dagger}U_{A_{j_2}}A_{j_2,2}U_{A_{j_2}}^{\dagger}\prod_{\substack{i=1\\i\ne j_1,j_2}}^{N}  X_{A_i'}\otimes\I_{A_i''}\ \tilde{\rho}_{0\ldots0}\right).
\end{eqnarray}
Now, using \eqref{consreal2} and \eqref{post-mea1}, we have that
\begin{eqnarray}
 \left\langle \tilde{A}_{1,1}A_{j_1,2}A_{j_2,2}\prod_{\substack{i=2\\i\ne j_1,j_2}}^{N} A_{i,1}\right\rangle=\Tr(r_{j_1,2}\rho_{j_1})\Tr(r_{j_2,2}\rho_{j_2})-\Tr(r_{j_1,3}\rho_{j_1})\Tr(r_{j_2,3}\rho_{j_2}).
\end{eqnarray}
Using the fact that $\Tr(r_{j_1,3}\rho_{j_1})=0$, we obtain that
\begin{eqnarray}\label{use3}
    \left\langle \tilde{A}_{1,1}A_{j_1,2}A_{j_2,2}\prod_{\substack{i=2\\i\ne j_1,j_2}}^{N} A_{i,1}\right\rangle=\Tr(r_{j_1,2}\rho_{j_1})\Tr(r_{j_2,2}\rho_{j_2}).
\end{eqnarray}
Similarly, computing the second term of \eqref{BE2} we obtain that
\begin{eqnarray}\label{use4}
    \left\langle A_{1,2}A_{j_1,2}\prod_{\substack{i=2\\i\ne j_1}}^{N} A_{i,2} \right\rangle=\Tr(r_{1,2}\rho_{1})\Tr(r_{j_1,2}\rho_{j_1}).
\end{eqnarray}
Adding all the terms from \eqref{use3} and \eqref{use4}, we obtain
\begin{eqnarray}\label{abcd}
    \mathcal{J}_N=-\frac{1}{N(N-1)}\sum_{\substack{j_1,j_2=1\\j_1\ne j_2}}^N\Tr(r_{j_1,2}\rho_{j_1})\Tr(r_{j_2,2}\rho_{j_2})
\end{eqnarray}
Employing then the fact that 
\begin{eqnarray}
\left[\sum_{j_1=1}^{N}\Tr(r_{j_1,2}\rho_{j_1})\right]^2=\sum_{j_1=1}^{N}[\Tr(r_{j_1,2}\rho_{j_1})]^2+\sum_{\substack{j_1,j_2=1\\j_1\ne j_2}}^N\Tr(r_{j_1,2}\rho_{j_1})\Tr(r_{j_2,2}\rho_{j_2})\geq 0,
\end{eqnarray}
one concludes that the expression \eqref{abcd} is upper bounded as
\begin{eqnarray}\label{use5}
 \mathcal{J}_N\leq\frac{1}{N(N-1)}\sum_{j_1=1}^{N}[\Tr(r_{j_1,2}\rho_{j_1})]^2.
\end{eqnarray}

Let us then observe that through the decomposition \eqref{consreal2} it follows from the fact that $A_{j_1,2}^2\leq\I$ that the matrices $r_{j_1,k}$ obey $\sum_kr_{j_1,k}^2\leq\I$, which directly implies that $r_{j_1,2}^2\leq\I$ for any $j_1$. Simultaneously, $\langle r_{j_1,2}^2\rangle\geq\langle r_{j_1,2}\rangle^2$ and consequently we have that $[\Tr(r_{j_1,2}\rho)]^2\leq\Tr(r_{j_1,2}^2\rho)\leq1$, which through \eqref{use5} implies finally that $\mathcal{J}_{N}\leq 1/(N-1)$, completing the proof.
\end{proof}

\section{Appendix B: Noisy scenario}
Let us proceed towards our result in the noisy scenario. For this purpose, we first need to find the robust self-testing statement of the states and measurements using the Bell inequality $\mathcal{I}_l$. Before proceeding, let us state an important Lemma that will be useful for the noisy self-testing proof. 

\begin{ulem}\label{lem1}
 Consider two matrices $A_{1,0},A_{1,1}$ that are hermitian and unitary with $\dl{\{A_{1,0},A_{1,1}\}\ket{\psi_{AB}}}\leq\varepsilon_N$,  where we additionally assume that $\varepsilon_N< 1$. Then, there exist a unitary transformation $U_1$ such that  
 \begin{eqnarray}\label{robumea11}
U_1A_{1,0}U_{1}^{\dagger}\ket{\psi'}=[(Z\otimes\I)\oplus \tilde{A}_{1,0}]\ket{\psi'},\qquad  \dl{(U_1A_{1,1}U_{1}^{\dagger}-[(X\otimes\I)\oplus \tilde{A}_{1,1}])\ket{\psi'}}\leq \varepsilon_N,
 \end{eqnarray}
%

%
 where $\ket{\psi'}= U_1\ket{\psi}$.
 \end{ulem}
 \begin{proof}
 To begin with, let us recall Jordan's Lemma, which states that two binary-outcome measurements can always be decomposed into blocks of size $2\times 2$ and $1\times1$ as
 \begin{eqnarray}\label{Ajor}
     A_{1,k}=\left(\sum_{i}A_{1,k}^{(i)}\otimes \proj{i}\right)\oplus \tilde{A}_{1,k}\qquad (k=0,1),
 \end{eqnarray}
 where $A_{1,k}^{(i)}$ are nontrivial $2\times 2$ blocks
 and $\tilde{A}_{1,k}$ is a matrix collecting all the $1\times 1$ blocks. 
 
Analogously, the state $\ket{\psi}$ decomposes as 
 \begin{eqnarray}\label{sjor}
     \ket{\psi}=\left(\sum_{i}\ket{\psi^{(i)}}\otimes \ket{i}\right)\oplus\ket{\tilde{\psi}}:=\ket{\overline{\psi}}\oplus\ket{\tilde{\psi}}.
 \end{eqnarray}
Consider now the relation $\dl{\{A_{1,0},A_{1,1}\}\ket{\psi}}\leq\varepsilon_N$ and use the above representation of the observables \eqref{Ajor} and state \eqref{sjor} as
 \begin{eqnarray}\label{eq112}
   \sum_{i}  \dl{\{A_{1,0}^{(i)},A_{1,1}^{(i)}\}\ket{\psi^{(i)}}}+\dl{\{\tilde{A}_{1,0},\tilde{A}_{1,1}\}\ket{\tilde{\psi}}}\leq\varepsilon_N.
 \end{eqnarray}
 Since $\tilde{A}_{1,0},\tilde{A}_{1,1}$ commute, along with the fact that they are unitary and Hermitian, allows us to conclude that
\begin{eqnarray}\label{eq111}
    \sum_{i}  \dl{\{A_{1,0}^{(i)},A_{1,1}^{(i)}\}\ket{\psi^{(i)}}}\leq\varepsilon_N,\qquad \dl{\ket{\tilde{\psi}}}\leq\frac{\varepsilon_N}{\sqrt{2}}\leq\varepsilon_N.
\end{eqnarray}
Notice that since $\ket{\psi}$ is normalised ($\dl{\ket{\psi}}=1$), 
\begin{equation}
 \sum_{i}\dl{\ket{\psi^{(i)}}}+\dl{\ket{\tilde{\psi}}}=1,
\end{equation}
which by virtue of \eqref{eq111} implies that 
\begin{eqnarray}\label{norm}
    \sum_{i}\dl{\ket{\psi^{(i)}}}\geq 1-\varepsilon_N.
\end{eqnarray}

Let us now observe that there exist unitary transformations $U_i$ such that
\begin{equation}
\hat{A}^{(i)}_{1,0}:=U_i\,A^{(i)}_{1,0}\,U_i^{\dagger}=Z,\qquad \hat{A}^{(i)}_{1,1}:=U_i\,A^{(i)}_{1,1}\,U_i^{\dagger}=\alpha_i Z+\beta_i X,    
\end{equation}
where $-1\leq \alpha_i,\beta_i\leq 1$ and $\alpha_i^2+\beta_i^2=1$ for every $i$.

Let us then define the three-dimensional real vectors $\vec{z}=(0,0,1)^T$ and $\vec{b_i}=(\beta_i,0,\alpha_i)^T$, as well as the vector of the Pauli matrices $\vec{\sigma}=(X,Y,Z)$. This allows us to represent 
\begin{equation}
\left\{\hat{A}_{1,0}^{(i)},\hat{A}_{1,1}^{(i)}\right\}=2\,\vec{z}\cdot\vec{b}_i\mathbbm{1},
\end{equation}
%
%
and also to conclude from the formula on the left in \eqref{eq111} that
\begin{eqnarray}
  \sum_{i}  \dl{\left\{\hat{A}_{1,0}^{(i)},\hat{A}_{1,1}^{(i)}\right\}\ket{\hat{\psi}^{(i)}}}=  2\sum_{i}  |\vec{z}\cdot \vec{b_i}|\dl{\ket{\psi^{(i)}}}\leq\varepsilon_N,
\end{eqnarray}
where $\ket{\hat{\psi}^{(i)}}=U_i \ket{\psi^{(i)}}$.

Let us now bound the distance between $A^{(i)}_{1,1}$ and the Pauli $X$ when acting on $\ket{\hat{\psi}^{(i)}}$. To this end, we first observe that 
%
\begin{eqnarray}\label{Tinto}
    \dl{\left(\hat{A}_{1,1}^{(i)}-X\right)\ket{\hat{\psi}^{(i)}}}^2= 2(1-\vec{x}\cdot\vec{b_i})\dl{\ket{\psi^{(i)}}}^2,
\end{eqnarray} 
where $\vec{x}=(1,0,0)^T$.

We then exploit that for any $i$, $\alpha_i^2+\beta_i^2=1$, which in terms of the vector $\vec{b}_i$ can be stated as
$(\vec{b_i}\cdot \vec{x})^2+(\vec{b_i}\cdot\vec{z})^2=1$.
This together with Eq. (\ref{Tinto}) gives
\begin{eqnarray}
    \dl{\left(\hat{A}_{1,1}^{(i)}-X\right)\ket{\hat{\psi}^{(i)}}}^2= 2\left[1-\sqrt{1-(\vec{b}_i\cdot\vec{z})^2}\right]\dl{\ket{\psi^{(i)}}}^2\leq 2(\vec{b}_i\cdot\vec{z})^2\dl{\ket{\psi^{(i)}}}^2.
\end{eqnarray}

Let us now consider the following unitary $U_1=[\sum_{i}\proj{i}\otimes U_i]\oplus \mathbbm{1}$, we can write
\begin{eqnarray}
\dl{[U_1\,A_{1,1}\,U_1^{\dagger}-(X\otimes\I)\oplus\tilde{A}_{1,1}]\ket{\psi'}}=\sum_{i}\dl{(\hat{A}_{1,1}^{(i)}-X)\ket{\hat{\psi}^{(i)}}}\leq2\sum_{i}|\vec{b_i}\cdot \vec{z}|\dl{\ket{\psi^{(i)}}}\leq \varepsilon_N.
\end{eqnarray}
This completes the proof. 
\end{proof}

Since the part of both observables $\tilde{A}_{1,k}$ composed of all the $1\times 1$ blocks do not give rise to any nonlocality, we omit them in further consideration. Let us restate
the above lemma without them as the following corollary.
\begin{cor}\label{cor1}
 Consider two matrices $A_{1,0},A_{1,1}$ that are hermitian and unitary with $\dl{\{A_{1,0},A_{1,1}\}\ket{\psi}}\leq\varepsilon_N$, where we additionally assume that $\varepsilon_N<1$. Then, there exists a unitary transformation $U_1$ such that  
 \begin{eqnarray}\label{cor11}
U_1\,A_{1,0}\,U_{1}^{\dagger}\ket{\psi'}=Z\otimes\I\ket{\psi'},\qquad  \dl{(U_1\, A_{1,1}\, U_{1}^{\dagger}-X\otimes\I)\ket{\psi'}}\leq 2\varepsilon_N
 \end{eqnarray}
%
%
 where $\ket{\psi'}= U_1\ket{\psi}$.
 \end{cor}
\begin{proof}The fact that the inequality (\ref{robumea11}) 
containg an additional factor $2$ as compared to 
(\ref{robumea11}) stems from the fact that $\ket{\overline{\psi}}$ is unnormalized and its norm is bounded as in (\ref{norm}) which needed to be taken into account.
\end{proof}

Let us now recall the Bell operator $\mathcal{I}_l$ as
\begin{equation}\label{BE1Nop}
\hat{\mathcal{I}}_{l}=(-1)^{l_1}\left[ (N-1)\tilde{A}_{1,1}\prod_{i=2}^N A_{i,1} +\sum_{i=2}^N(-1)^{l_i}\tilde{A}_{1,0}A_{i,0}\right],
\end{equation}
where $l\equiv l_1\ldots l_N$ such that $l_1,l_2,\ldots,l_N=0,1$ is the binary representation of the outcome $l$.
Let us first introduce the sum-of-squares decomposition of the Bell operator \eqref{BE1Nop}
, for which we define $\mathcal{J}_{l}=\beta_Q\I-\mathcal{\hat{I}}_{l}$. Then
\begin{eqnarray}\label{SOSnew1}
   2\beta_Q\mathcal{J}_{l}=\mathcal{J}^2_{l}+\sum_{i,j=2}^NQ_{i,j,l_i,l_j}^2+(N-1)\sum_{j=2}^NT_{j,l_j}^2
\end{eqnarray}
where
\begin{eqnarray}\label{SOSnew2}
   Q_{i,j,l_i,l_j}=\tilde{A}_{1,0}\left[(-1)^{l_i}A_{i,0}-(-1)^{l_j}A_{j,0}\right],
\end{eqnarray}
and
\begin{eqnarray}\label{SOSnew3}
   T_{j,l_j}=\tilde{A}_{1,1}A_{j,0}\prod_{i\ne j,1}^NA_{i,1}+(-1)^{l_j}\tilde{A}_{1,0}A_{j,1}
\end{eqnarray}
such that $\tilde{A}_{1,1}=(A_{1,0}+A_{1,1})/\sqrt{2}$ and $\tilde{A}_{1,0}=(A_{1,0}-A_{1,1})/\sqrt{2}$ with $\beta_Q=2(N-1)$. We are now ready to present the robust self-testing of states and measurements using the Bell functional $\mathcal{I}_{l}$. As the dimension is unrestricted, we assume that the measurements of all parties are projective.

Another SOS decomposition of the Bell operator \eqref{BE1Nop} which will be useful for robust self-testing is given by
\begin{eqnarray}\label{SOS1}
   2\left(\beta_Q\I-\mathcal{\hat{I}}_{1,l_1l_2\ldots l_N}\right)=(N-1)P_{1,l_1}^2+\sum_{i=2}^NP_{i,l_1,l_i}^2
\end{eqnarray}
where $\beta_Q=2(N-1)$ and
\begin{eqnarray}\label{SOS2}
   P_{1,l_1}&=&\I-(-1)^{l_1}\tilde{A}_{1,1}\prod_{i=2}^N A_{i,1},
   \qquad P_{i,l_1,l_i}=\I-(-1)^{l_1+l_i}\tilde{A}_{1,0} A_{i,0}\qquad i=2,3,\ldots,N.
\end{eqnarray}
We are now ready to present the robust self-testing of states and measurements using the Bell functional $\mathcal{I}_{l}$. As the dimension is unrestricted, we assume that the measurements of all parties are projective.

\begin{thm}
   Suppose that the the Bell functionals $\langle\hat{\mathcal{I}}_{l}\rangle$ \eqref{BE1Nop} attains a value $\varepsilon_N-$close to the quantum bound, that is,
\begin{eqnarray}
 \langle\psi_l|\hat{\mathcal{I}}_{l}|\psi_l\rangle\geq 2(N-1)-\varepsilon_N\qquad \forall l.
\end{eqnarray}
along with the probability of the outcomes of the measurement of Bob being $|\overline{P}(l)-1/2^N|\leq\varepsilon_N$  for all $l$. 
Then, there exists unitaries $U_i:\mathcal{H}_{A_i}\rightarrow\mathcal{H}_{A_i}$
for $i=1,\ldots,N$ such that the state is certified as
\begin{eqnarray}\label{Rob1manu1}
 \dl{\ket{\tilde{\psi}_l}-\ket{\phi_l}\ket{\xi_l}}\leq \left(\delta_N+\sqrt{2(N-1)}\right)\sqrt{2\varepsilon_N}.
\end{eqnarray}
where $\ket{\tilde{\psi}_l}=\bigotimes_{i=1}^NU_{i}\ket{\psi_l}$ and $\delta_N=16(N-1)+ 2N(N-1)\left[\sqrt{2}+1+\sqrt{\frac{1}{N-1}}\right]$, while the measurements are characterised as
\begin{eqnarray}
    \dl{U_1A_{1,0}U_1^{\dagger}\ket{\tilde{\psi}_l}-\frac{X+Z}{\sqrt{2}}\otimes\I\ket{\phi_l}\ket{\xi_l}}&\leq& \left(\delta_N+\sqrt{2(N-1)}\right)\sqrt{2\varepsilon_N},\nonumber\\ \dl{U_1A_{1,1}U_1^{\dagger}\ket{\tilde{\psi}_l}-\frac{X-Z}{\sqrt{2}}\ket{\phi_l}\ket{\xi_l}}&\leq& \left(8+\delta_N+\sqrt{2(N-1)}\right)\sqrt{2\varepsilon_N},\nonumber\\
    \dl{U_iA_{i,0}U_i^{\dagger}\ket{\tilde{\psi_l}}-Z\otimes\I\ket{\phi_l}\ket{\xi_l}}&\leq& \left(\delta_N+\sqrt{2(N-1)}\right)\sqrt{2\varepsilon_N},\quad (i=2,\ldots,N)\nonumber\\ \dl{U_iA_{i,1}U_i^{\dagger}\ket{\tilde{\psi_l}}-X\otimes\I\ket{\phi_l}\ket{\xi_l}}&\leq& \left(\sqrt{2}+1+\sqrt{\frac{1}{N-1}}+\delta_N+\sqrt{2(N-1)}\right)\sqrt{2\varepsilon_N},\quad (i=2,\ldots,N).\quad
\end{eqnarray}
\end{thm}
\begin{proof}

 Before proceeding, let us recall that the state shared among the external parties $A_i$ when the central party $B$ gets an outcome $l$ is given by 
\begin{eqnarray}\label{54}
\rho^l_{{A_{1,1}\ldots A_N}}=\frac{1}{\overline{P}(l)}\Tr_{\overline{A}_1\ldots\overline{A}_N}\left[\left(\I_{A_{1,1}\ldots A_N}\otimes R_l\right)\prod_{i=1}^N\proj{\psi_{A_i\overline{A}_i}}\right].
\end{eqnarray}
Now, let us purify this state $\rho^l_{{A_{1,1}\ldots A_N}}$ by adding an ancillary system $G$ such that 
\begin{eqnarray}\label{puristate1}
   \rho^l_{{A_{1,1}\ldots A_N}}=\Tr_G\left(\proj{\psi^l}_{{A_{1,1}\ldots A_NG}}\right).
\end{eqnarray}
For simplicity, we will drop the subscript from the above state and also the index $l$ and recall them back at the end of the proof. 

Let us begin the proof by first noticing from \cite{Bamps} that if the Bell inequality $\mathcal{I}_l$ is violated $\varepsilon_N-$close to the quantum bound, that is, $\langle\hat{\mathcal{I}}_l\rangle\geq\beta_Q-\varepsilon_N$, then the terms in its SOS decomposition, that is, $\beta_Q\I-\hat{\mathcal{I}}_l=1/2\sum_{i}P_i^{\dagger}P_i$, is bounded as
$||P_i\ket{\psi_l}||\leq\sqrt{2\varepsilon_N}$. Thus, if the Bell inequality $\mathcal{I}_{l}$ is violated $\varepsilon_N-$close to the quantum bound then from the SOS decompositions \eqref{SOSnew1} and \eqref{SOSnew2}, we have
\begin{eqnarray}\label{robustSOS1}
   \| P_{1,l_1}\ket{\psi_l}\|\leq\sqrt{\frac{2\varepsilon_N}{N-1}}
 \end{eqnarray}
 and,
 \begin{eqnarray}\label{robustSOS2}
      \|P_{i,l_1,l_i}\ket{\psi_l}\|\leq\sqrt{2\varepsilon_N}
 \end{eqnarray}
and from \eqref{SOSnew1}, we have
\begin{eqnarray}\label{robustSOS3}
   \|T_{j,l_j}\ket{\psi_l}\|\leq2\sqrt{\varepsilon_N},
\end{eqnarray}
and
\begin{eqnarray}\label{robustSOS4}
    \| Q_{i,j,l_i,l_j}\ket{\psi_l}\|\leq2\sqrt{(N-1)\varepsilon_N}, \qquad \|\mathcal{J}_l\ket{\psi_l}\|\leq2\sqrt{(N-1)\varepsilon_N}.
\end{eqnarray}

Now, let us find some relations required for finding robustness of the state. First, expanding \eqref{robustSOS2} using \eqref{SOS2} and then recalling that $A_{i,0}$ are unitary and $A_{i,j}^2=\I$ for any $i$, we obtain that
\begin{eqnarray}\label{robustrel1111}
    \dl{\left((-1)^{l_i+l_1}A_{i,0}-\tilde{A}_{1,0}\right)\ket{\psi_l}}\leq\sqrt{2\varepsilon_N}\qquad i=2,\ldots,N.
\end{eqnarray}
Considering the above relation \eqref{robustrel1111} and then recalling that $A_{i,0}$ are hermitian and unitary, we obtain that
\begin{eqnarray}\label{robustrel1}
    \dl{(\I-(-1)^{l_i+l_1}\tilde{A}_{1,0}A_{i,0})\ket{\psi_l}}\leq\sqrt{2\varepsilon_N}
\end{eqnarray}
which on utilising triangle inequality can be expressed as
\begin{eqnarray}
    \dl{(\I-\tilde{A}_{1,0}^2)\ket{\psi_l}}\leq\sqrt{2\varepsilon_N}+ \dl{\tilde{A}_{1,0}(\tilde{A}_{1,0}-(-1)^{l_i+l_1}A_{i,0})\ket{\psi_l}}&\leq&\sqrt{2\varepsilon_N}+\sqrt{2}\ \dl{(\tilde{A}_{1,0}-(-1)^{l_i+l_1}A_{i,0})\ket{\psi_l}}\nonumber\\ &\leq& (1+\sqrt{2})\sqrt{2\varepsilon_N}
\end{eqnarray}
where we used the fact that $\tilde{A}_{1,0}\leq\sqrt{2}\I$. Consequently, we have from the above expression that
\begin{eqnarray}\label{noisyac1}
 \dl{\{A_{1,0},A_{1,1}\}\ket{\psi_l}}\leq (2+\sqrt{2})\sqrt{2\varepsilon_N}\leq 4\sqrt{2\varepsilon_N}.
\end{eqnarray}
Considering Corollary \ref{cor1}, we thus obtain that there exists a unitary $V_A$ such that
\begin{eqnarray}
   V_AA_{1,0}V_{A}^{\dagger}\ket{\psi''_l}=Z\otimes\I\ket{\psi''_l},\qquad  \dl{(V_AA_{1,1}V_{A}^{\dagger}-X\otimes\I)\ket{\psi''_l}}\leq 8\sqrt{2\varepsilon_N}
 \end{eqnarray}
 where $\ket{\psi''}=V_A\ket{\psi_l}$. Notice that we can always find a unitary $\mathcal{U}$ such that $\mathcal{U}Z\mathcal{U}^{\dagger}=(X+Z)/\sqrt{2}$ and $\mathcal{U}X\mathcal{U}^{\dagger}=(X-Z)/\sqrt{2}$. Consequently, we have that
 \begin{eqnarray}\label{roburel15}
   U_1A_{1,0}U_{A}^{\dagger}\ket{\psi'_l}=\frac{X+Z}{\sqrt{2}}\otimes\I\ket{\psi'_l},\qquad  \dl{\left(U_1A_{1,1}U_{A}^{\dagger}-\frac{X-Z}{\sqrt{2}}\otimes\I\right)\ket{\psi'_l}}\leq 8\sqrt{2\varepsilon_N}
 \end{eqnarray}
 where $U_1=\mathcal{U}\otimes\I V_A$ and $\ket{\psi'_l}=U_1\ket{\psi_l}$. Furthermore, from the above condition, we also obtain that
\begin{eqnarray}\label{roburel11}
     \dl{U_1\tilde{A}_{1,0}U_1^{\dagger}\ket{\psi'_l}-Z\otimes\I \ket{\psi'_l}}\leq 8\sqrt{2\varepsilon_N},\qquad \dl{U_1\tilde{A}_{1,1}U_1^\dagger\ket{\psi'_l}-X\otimes\I \ket{\psi'_l}}\leq 8\sqrt{2\varepsilon_N}.
\end{eqnarray}


Now, expanding \eqref{robustSOS1} using \eqref{SOS2} and then recalling that $A_{i,0}$ are unitary and $A_i^2=\I$ for any $i$, we obtain that
\begin{eqnarray}\label{robustrel2}
    \dl{\left((-1)^{l_1}\tilde{A}_{1,1}-A_{j,1}\prod_{i\ne j,1}^NA_{i,1}\right)\ket{\psi_l}}\leq\sqrt{\frac{2\varepsilon_N}{N-1}}\qquad j=2,\ldots,N.
\end{eqnarray}
Next, we find the following relation for $j=2,\ldots,N$
\begin{equation}\label{robustrel4}
    \dl{\left(A_{j,0}A_{j,1}+A_{j,1}A_{j,0}\right)\ket{\psi_l}}=\dl{\left[(-1)^{l_1}T_{j,l_j}-A_{j,0}\left((-1)^{l_1}\tilde{A}_{1,1}\prod_{i\ne j,1}^NA_{i,1}-A_{j,1}\right)-A_{j,1}\left((-1)^{l_1+l_j}\tilde{A}_{1,0}-A_{j,0}\right)\right]\ket{\psi_l}}
\end{equation}
where $T_{j,l_j}$ is given in \eqref{SOSnew3}. Now, using the triangle inequality, we find that
\begin{eqnarray}
    \dl{\left(A_{j,0}A_{j,1}+A_{j,1}A_{j,0}\right)\ket{\psi_l}}\leq\dl{(-1)^{l_1}T_{j,l_j}\ket{\psi_l}}&+&\dl{A_{j,0}\left((-1)^{l_1}\tilde{A}_{1,1}\prod_{i\ne j,1}^NA_{i,1}-A_{j,1}\right)\ket{\psi_l}}\nonumber\\&+&\dl{A_{j,1}\left((-1)^{l_j+l_1}\tilde{A}_{1,0}-A_{j,0}\right)\ket{\psi_l}}.
\end{eqnarray}
Then recalling that $A_{i,j}$ are unitary and $A_{i,j}^2=\I$ for any $i$, we obtain using \eqref{robustrel1}, \eqref{robustrel2} and \eqref{robustSOS2} that
\begin{eqnarray}
     \dl{\left(A_{j,0}A_{j,1}+A_{j,1}A_{j,0}\right)\ket{\psi_l}}\leq\left[\sqrt{2}+1+\sqrt{\frac{1}{N-1}}\right]\sqrt{2\varepsilon_N}\qquad j=2,\ldots,N.
\end{eqnarray}
Again considering Corrollary \ref{cor1}, we obtain that
\begin{eqnarray}\label{roburel12}
    U_iA_{i,0}U_i^\dagger\ket{\psi'_l}=Z\otimes\I\ket{\psi'_l},\qquad\dl{U_iA_{i,1}U_i^\dagger\ket{\psi'_l}-X\otimes\I\ket{\psi'_l}}\leq \left[\sqrt{2}+1+\sqrt{\frac{1}{N-1}}\right]\sqrt{2\varepsilon_N}.
\end{eqnarray}

Now, we consider \eqref{roburel11} and using the fact that $A_{i,1}$ is unitary, we obtain that
\begin{eqnarray}
    \dl{U_1\otimes U_i\tilde{A}_{1,1}\otimes A_{i,1}\ket{\psi_l}-X\otimes\I\otimes (U_i A_{i,1} U_i^{\dagger}) (U_1\otimes U_i\ket{\psi_l})}\leq 8\sqrt{2\varepsilon_N}.
\end{eqnarray}
Using triangle inequality, the above expression can be expressed as
\begin{eqnarray}\label{roburel13}
   \dl{U_1\otimes U_i\tilde{A}_{1,1}\otimes A_{i,1}\ket{\psi_l}-X_{A_1}\otimes\I_{A_1''}\otimes X_{A_i'}\otimes\I_{A_i''}(U_1\otimes U_i\ket{\psi_l})}&\leq& 8\sqrt{2\varepsilon_N}+\dl{U_iA_{i,1}\ket{\psi_l}-X\otimes\I (U_i\ket{\psi_l})}\nonumber\\ &\leq& 8\sqrt{2\varepsilon_N}+\left[\sqrt{2}+1+\sqrt{\frac{1}{N-1}}\right]\sqrt{2\varepsilon_N}. 
\end{eqnarray}
Continuing similarly, we obtain that
\begin{eqnarray}
    \dl{ U_1\tilde{A}_{1,1}\prod_{i=2}^N U_iA_{i,1}\ket{\psi_l}-\prod_{i=1}^N X_{A_i}\otimes\I_{A_i''}\ket{\tilde{\psi}_l}}&\leq& 8\sqrt{2\varepsilon_N}+ (N-1)\left[\sqrt{2}+1+\sqrt{\frac{1}{N-1}}\right]\sqrt{2\varepsilon_N}
\end{eqnarray}
where $\ket{\tilde{\psi}_l}=\prod_iU_i\ket{\psi_l}$.
Similarly, we can also obtain from \eqref{roburel11} that for all $i=2,\ldots,N$
\begin{eqnarray}\label{roburel14}
     \dl{U_1\tilde{A}_{1,0} U_i A_{i,0}\ket{\psi_l}-Z_{A_1}\otimes\I_{A_1''}\otimes Z_{A_i'}\otimes\I_{A_i''}\ket{\tilde{\psi}_l}}\leq8\sqrt{2\varepsilon_N}+\left[\sqrt{2}+1+\sqrt{\frac{1}{N-1}}\right]\sqrt{2\varepsilon_N}. 
\end{eqnarray}
Adding the above two formulas \eqref{roburel13} and \eqref{roburel14} with proper factors in front of them, we obtain
\begin{eqnarray}
    \dl{(\hat{\mathcal{I}}_l'-\hat{\mathcal{I}}_l^{\mathrm{id}}\otimes\I_{A_1''\ldots A_N''})\ket{\tilde{\psi_l}}}\leq  16(N-1)\sqrt{2\varepsilon_N}+ N(N-1)\left[\sqrt{2}+1+\sqrt{\frac{1}{N-1}}\right]\sqrt{2\varepsilon_N},
\end{eqnarray}
where $\hat{\mathcal{I}}_l'=(\prod_iU_i) \hat{\mathcal{I}}_l(\prod_iU_i)^{\dagger} $ with $\hat{\mathcal{I}}^{\mathrm{id}}=(\prod_{i=1}^NX_{A_i'}+\sum_{i=2}^N Z_{A_1'} Z_{A_i'})$. For simplicity, we will now denote 
\begin{equation}
 \delta_N\equiv16(N-1)+ 2N(N-1)\left[\sqrt{2}+1+\sqrt{\frac{1}{N-1}}\right].   
\end{equation}

Now, using $\dl{\mathcal{J}_l\ket{\psi_l}}\leq\sqrt{2\beta_Q\varepsilon_N}$ and then triangle inequality, we obtain that
\begin{eqnarray}\label{55}
    \dl{(\beta_Q\I-\mathcal{\hat{I}}_l^{\mathrm{id}}\otimes\I_{A_1''\ldots A_N''})\ket{\tilde{\psi_l}}}\leq \left(\delta_N+\sqrt{2(N-1)}\right)\sqrt{2\varepsilon_N}. 
\end{eqnarray}
Notice that $\mathcal{\hat{I}}_l^{\mathrm{id}}\ket{\phi_s}=\lambda_s\ket{\phi_s}$ with $\ket{\phi_s}$ given in eq. 6 of the manuscript and $\lambda_s= (-1)^{l_1+s_1}[(N-1)+\sum_{i=2}^N(-1)^{l_i+s_i}]$. Consequently, the eigenspace of $\mathcal{\hat{I}}_l^{\mathrm{id}}$ is spanned by the GHZ-like vectors. Let us now express the above formula \eqref{55} using the eigenspace of $\mathcal{\hat{I}}_l^{\mathrm{id}}$ as
\begin{eqnarray}
     \dl{\left[\sum_{s=1}^{2^N}(\beta_Q-\lambda_s)\proj{\phi_s}\otimes\I_{A_1''\ldots A_N''}\right]\ket{\tilde{\psi_l}}}\leq \left(\delta_N+\sqrt{2(N-1)}\right)\sqrt{2\varepsilon_N}. 
\end{eqnarray}
Expanding $\ket{\tilde{\psi_l}}=\sum_{i=1}^{2^N}\alpha_i\ket{\phi_i}\ket{\xi_i}$ with $0\leq\alpha_i\leq1$ and $\ket{\xi_i}\in \bigotimes_{i=1}^N\mathcal{H}_{A''_i}$, we obtain from the above expression that
\begin{eqnarray}
 \dl{\sum_{s=1}^{2^N}(\beta_Q-\lambda_s)\alpha_s\ket{\phi_s}\ket{\xi_s}}\leq \left(\delta_N+\sqrt{2(N-1)}\right)\sqrt{2\varepsilon_N}. 
\end{eqnarray}
which on expanding the left-hand side results in
\begin{eqnarray}
  \sqrt{\sum_{s=1}^{2^N}\alpha_s^2(\beta_Q-\lambda_s)^2}\leq \left(\delta_N+\sqrt{2(N-1)}\right)\sqrt{2\varepsilon_N}.
\end{eqnarray}
It is simple to observe that the difference between the highest and second highest eigenvalue of $\mathcal{\hat{I}}_l^{\mathrm{id}}$ is $2$ and thus $\min_{s\ne l}|\beta_Q-\lambda_s|=2$. Consequently, from the above formula we have that
\begin{eqnarray}\label{robu32}
    \sum_{\substack{s=1\\s\ne l}}^{2^N}\alpha_s^2\leq \left(\delta_N+\sqrt{2(N-1)}\right)^2\frac{\varepsilon_N}{2}.
\end{eqnarray}
As $\alpha_s\leq1$, we have that $\alpha_{l}\geq\alpha_{l}^2\geq1-\left(\delta_N+\sqrt{2(N-1)}\right)^2\frac{\varepsilon_N}{2}.$ and thus, 
\begin{eqnarray}\label{65}
    \langle\tilde{\psi_l}|(\ket{\phi_l}\ket{\xi_l})=\alpha_l\geq 1-\left(\delta_N+\sqrt{2(N-1)}\right)^2\frac{\varepsilon_N}{2}.
\end{eqnarray}
Consequently, we have that
\begin{eqnarray}
    \dl{\ket{\tilde{\psi_l}}-\ket{\phi_l}\ket{\xi_l}}=\sqrt{2[1-\mathrm{Re}( \langle\tilde{\psi_l}|(\ket{\phi_l}\ket{\xi_l}))]}\leq 
    \left(\delta_N+\sqrt{2(N-1)}\right)\sqrt{\varepsilon_N}.
\end{eqnarray}
Let us also compute the following quantity 
\begin{eqnarray}\label{62}
    \dl{\proj{\tilde{\psi_l}}-\proj{\phi_l}\otimes\proj{\xi_l}}_{1}=2\sqrt{1-|\langle\tilde{\psi_l}|(\ket{\phi_l}\ket{\xi_l})|^2},
\end{eqnarray}
where $\|\cdot\|_1$ stands for the trace norm.
Consequently, by virtue of the from Eq. \eqref{65} we obtain
\begin{eqnarray}
      \dl{\proj{\tilde{\psi_l}}-\proj{\phi_l}\otimes\proj{\xi_l}}_{1}&\leq&2\sqrt{1-\left(1-\left(\delta_N+\sqrt{2(N-1)}\right)^2\frac{\varepsilon_N}{2}\right)^2}\nonumber\\
      &\leq& 2\left(\delta_N+\sqrt{2(N-1)}\right)\sqrt{\varepsilon_N}.
\end{eqnarray}
where, for simplicity of the above expression, we ignored the negative term in the right-hand side.
This completes the proof of the robustness of the state to experimental errors. 

Finally, from \eqref{roburel15} using triangle inequality, we can obtain that
\begin{eqnarray}
      \dl{U_1A_{1,1}U_1^\dagger\ket{\tilde{\psi_l}}-\frac{X-Z}{\sqrt{2}}\otimes\I \ket{\phi_l}\ket{\xi_l}}\leq 8\sqrt{2\varepsilon_N}+ \dl{\frac{X-Z}{\sqrt{2}}\otimes\I (\ket{\phi_l}\ket{\xi_1}-\ket{\tilde{\psi_l}})}\leq   \left(8+\delta_N+\sqrt{2(N-1)}\right)\sqrt{2\varepsilon_N}\nonumber\\
\end{eqnarray}
and 
\begin{eqnarray}
    \dl{U_1A_{1,0}U_1^\dagger\ket{\tilde{\psi_l}}-\frac{X+Z}{\sqrt{2}}\otimes\I \ket{\phi_l}\ket{\xi_l}}\leq  \dl{\frac{X+Z}{\sqrt{2}}\otimes\I (\ket{\phi_l}\ket{\xi_1}-\ket{\tilde{\psi_l}})}\leq   \left(\delta_N+\sqrt{2(N-1)}\right) \sqrt{2\varepsilon_N}
\end{eqnarray}
Similarly, from \eqref{roburel12} using triangle inequality, we can obtain that
\begin{eqnarray}
      \dl{U_iA_{i,1}U_i^{\dagger}\ket{\tilde{\psi_l}}-X\otimes\I \ket{\phi_l}\ket{\xi_l}}\leq   \left(\sqrt{2}+1+\sqrt{\frac{1}{N-1}}+\delta_N+\sqrt{2(N-1)}\right)\sqrt{2\varepsilon_N}
\end{eqnarray}
and 
\begin{eqnarray}
      \dl{U_iA_{i,0}U_i^{\dagger}\ket{\tilde{\psi_l}}-Z\otimes\I \ket{\phi_l}\ket{\xi_l}}\leq   \left(\delta_N+\sqrt{2(N-1)}\right)\sqrt{2\varepsilon_N}.
\end{eqnarray}

Let us now show that the $\ket{\xi_l}$ are close to some product state. For this purpose, we assume $\ket{\xi_l}=\ket{\xi}$ for all $l$ and thus $\Tr_G\proj{\xi_l}=\sigma_0$. This holds well for the case when $\varepsilon_N=0$ [see Ref. \cite{sarkar2023universal}]. Now, considering \eqref{62} and tracing over the system $G$ gives us
\begin{eqnarray}
      \dl{\Tr_G\proj{\tilde{\psi_l}}-\proj{\phi_l}\otimes\sigma_0}_1\leq 
    \left(\delta_N+\sqrt{2(N-1)}\right)\sqrt{\varepsilon_N}.
\end{eqnarray}
Now, expanding the left-hand side using \eqref{54} and then summing over $l$ gives us
\begin{eqnarray}
    \dl{\prod_{i=1}^{2^N}\rho_{A_i}-\sum_{l=1}^{2^N}\overline{P}(l)\proj{\phi_l}\otimes\sigma_0}_1\leq \overline{P}(l)
    \left(\delta_N+\sqrt{2(N-1)}\right)\sqrt{\varepsilon_N}\leq  \left(\delta_N+\sqrt{2(N-1)}\right)\sqrt{\varepsilon_N}
\end{eqnarray}
and thus taking a partial trace over $A_1'\ldots A_N'$ in the left-hand-side of the above formula gives us
\begin{eqnarray}\label{prodclose}
     \dl{\prod_{i=1}^{2^N}\sigma_{A_i''}-\sigma_0}_1\leq  \left(\delta_N+\sqrt{2(N-1)}\right)\sqrt{\varepsilon_N}
\end{eqnarray}
where $\sigma_{A_i''}=\Tr_{A_i'}\rho_{A_i}$.
This will be particularly useful for finding the maximal value of $\mathcal{J}_N$ \eqref{BE2} for RQT.
This completes the proof.
\end{proof}
Let us now compute the maximal value of $\mathcal{J}_N$ \eqref{BE2} for RQT.
\begin{thm}
     Consider the scenario depicted in Fig. 1 of the manuscript and suppose that the Bell functionals $\langle\hat{\mathcal{I}}_{l}\rangle$ \eqref{BE1Nop} attains a value $\varepsilon_N-$close to the quantum bound 
along with the probability of the outcomes of the measurement of Bob being $|\overline{P}(l)-1/2^N|\leq\varepsilon_N$  for all $l$. 
Then, the maximal value of $\mathcal{J}_N$ \eqref{BE2} for RQT is upper bounded as
\begin{eqnarray}
\beta_{\mathrm{RQT}}\leq \frac{1}{N-1}+ f(N)\sqrt{2\varepsilon_N}
\end{eqnarray}
where $f(N)=\left[8+(N-3)\left(\sqrt{2}+1+\sqrt{\frac{1}{N-1}}\right)\right]+2\left(\delta_N+\sqrt{2(N-1)}\right)+\frac{N^2}{2}\left(\delta_N+\sqrt{2(N-1)}\right)^2$.
\end{thm}
\begin{proof} Let us begin by considering the relation \eqref{roburel11} for $l=00\ldots0\equiv 0$ and multiply it with the observables $A'_{j_1,2}A'_{j_2,2}\prod_{\substack{i=2\\i\ne j_1,j_2}}^{N} A'_{i,1}$ where $A'_{i,j}=U_iA_{i,j}U_i^{\dagger}$ to obtain
\begin{eqnarray}
  \dl{\tilde{A}'_{1,1}A'_{j_1,2}A'_{j_2,2}\prod_{\substack{i=2\\i\ne j_1,j_2}}^{N}A'_{i,1}\ket{\psi'_0}-X\otimes\I\otimes A'_{j_1,2}A'_{j_2,2}\prod_{\substack{i=2\\i\ne j_1,j_2}}^{N}A'_{i,1} \ket{\psi'_0}}\leq 8\sqrt{2\varepsilon_N}.
\end{eqnarray}
Now, using triangle inequality and \eqref{roburel12} we obtain that
\begin{eqnarray}\label{73}
  \dl{\tilde{A}'_{1,1}A'_{j_1,2}A'_{j_2,2}\prod_{\substack{i=2\\i\ne j_1,j_2}}^{N}A'_{i,1}\ket{\psi'_0}-\prod_{\substack{i=1\\i\ne j_1,j_2}}^{N}X_{A_i'}\otimes\I_{A_i''}\otimes A'_{j_1,2}A'_{j_2,2} \ket{\psi'_0}}\leq \left[8+(N-3)\left(\sqrt{2}+1+\sqrt{\frac{1}{N-1}}\right)\right]\sqrt{2\varepsilon_N}.\nonumber\\
\end{eqnarray}
Similarly, considering \eqref{roburel12} for some $i=k$ and multiplying it with $A_{1,2}'A_{j_1,2}'\prod_{\substack{i=2\\i\ne j_1,k}}^{N} A_{i,2}'$ and then using triangle inequality we can obtain 
\begin{eqnarray}\label{74}
  \dl{\tilde{A}'_{1,2}A'_{j_1,2}\prod_{\substack{i=2\\i\ne j_1}}^{N}A'_{i,1}\ket{\psi'_0}-\prod_{\substack{i=2\\i\ne j_1}}^{N}X_{A_i'}\otimes\I_{A_i''}\otimes A'_{1,2}A'_{j_1,2} \ket{\psi'_0}}\leq   (N-2)\left[\sqrt{2}+1+\sqrt{\frac{1}{N-1}}\right]\sqrt{2\varepsilon_N}.\ \ 
\end{eqnarray}
Considering now \eqref{73} and using Cauchy-Schwartz inequality we obtain
\begin{eqnarray}
   \left|\bra{\psi'_0}\tilde{A}'_{1,1}A'_{j_1,2}A'_{j_2,2}\prod_{\substack{i=2\\i\ne j_1,j_2}}^{N}A'_{i,1}\ket{\psi'_0}-\bra{\psi'_0}\prod_{\substack{i=1\\i\ne j_1,j_2}}^{N}X_{A_i'}\otimes\I_{A_i''}\otimes A'_{j_1,2}A'_{j_2,2} \ket{\psi'_0}\right|\leq \left[8+(N-3)\left[\sqrt{2}+1+\sqrt{\frac{1}{N-1}}\right]\right]\sqrt{2\varepsilon_N}.\nonumber\\
\end{eqnarray}
Again, using triangle inequality we obtain that
\begin{eqnarray}
   \left|\bra{\psi'_0}\tilde{A}'_{1,1}A'_{j_1,2}A'_{j_2,2}\prod_{\substack{i=2\\i\ne j_1,j_2}}^{N}A'_{i,1}\ket{\psi'_0}-\bra{\psi'_0}\prod_{\substack{i=1\\i\ne j_1,j_2}}^{N}X_{A_i'}\otimes\I_{A_i''}\otimes A'_{j_1,2}A'_{j_2,2} \ket{\phi_0}\ket{\xi_0}\right|\nonumber\\ \leq \left[8+(N-3)\left(\sqrt{2}+1+\sqrt{\frac{1}{N-1}}\right)\right]\sqrt{2\varepsilon_N}+\dl{\ket{\psi'_0}-\ket{\phi_0}\ket{\xi_0}}
\end{eqnarray}
which again using triangle inequality and the fact that $\dl{M}=\dl{M^{\dagger}}$ can be written as
\begin{eqnarray}
   \left|\bra{\psi'_0}\tilde{A}'_{1,1}A'_{j_1,2}A'_{j_2,2}\prod_{\substack{i=2\\i\ne j_1,j_2}}^{N}A'_{i,1}\ket{\psi'_0}-\bra{\xi_0}\bra{\phi_0}\prod_{\substack{i=1\\i\ne j_1,j_2}}^{N}X_{A_i'}\otimes\I_{A_i''}\otimes A'_{j_1,2}A'_{j_2,2} \ket{\phi_0}\ket{\xi_0}\right|\nonumber\\ \leq \left[8+(N-3)\left(\sqrt{2}+1+\sqrt{\frac{1}{N-1}}\right)\right]\sqrt{2\varepsilon_N}+2\dl{\ket{\psi'_0}-\ket{\phi_0}\ket{\xi_0}}.
\end{eqnarray}
Now, using the fact that $\Tr(RQ)\leq\dl{R}_{\infty}\dl{Q}_1$ and the junk state $\ket{\xi_0}$ is close to the product state as in \eqref{prodclose} and thus we have that 
\begin{eqnarray}\label{78}
\left|\langle\tilde{A}'_{1,1}A'_{j_1,2}A'_{j_2,2}\prod_{\substack{i=2\\i\ne j_1,j_2}}^{N}A'_{i,1}\rangle-\Tr\left(\prod_{\substack{i=1\\i\ne j_1,j_2}}^{N}X_{A_i'}\otimes\I_{A_i''}\otimes A'_{j_1,2}A'_{j_2,2} \proj{\phi_0}\bigotimes_{i=1}^N\sigma_{A_i''}\right)\right| \leq f(N)\sqrt{2\varepsilon_N}
\end{eqnarray}
where $f(N)=\left[8+(N-3)\left(\sqrt{2}+1+\sqrt{\frac{1}{N-1}}\right)\right]+3\left(\delta_N+\sqrt{2(N-1)}\right)$ where $\delta_N=16(N-1)+ 2N(N-1)\left[\sqrt{2}+1+\sqrt{\frac{1}{N-1}}\right]$. Notice that in the above formula, we used the fact that $A'_{i,2}$ are observables with eigenvalues $+1,-1$ and thus 
\begin{eqnarray}
    \dl{\prod_{\substack{i=1\\i\ne j_1,j_2}}^{N}X_{A_i'}\otimes\I_{A_i''}\otimes A'_{j_1,2}A'_{j_2,2}}_{\infty}=1.
\end{eqnarray}
Similarly, from \eqref{74} we have that
\begin{eqnarray}\label{79}
\left|\langle\tilde{A}'_{1,2}A'_{j_1,2}\prod_{\substack{i=2\\i\ne j_1}}^{N}A'_{i,1}\rangle-\Tr\left(\prod_{\substack{i=2\\i\ne j_1}}^{N}X_{A_i'}\otimes\I_{A_i''}\otimes A'_{1,2}A'_{j_1,2} \proj{\phi_0}\bigotimes_{i=1}^N\sigma_{A_i''}\right)\right| \leq f(N)\sqrt{2\varepsilon_N}
\end{eqnarray}
where we use the fact that $\sqrt{2}+1+\frac{1}{\sqrt{N-1}}\leq8$. 

Adding Eqs. \eqref{78} and \eqref{79} we obtain that
\begin{eqnarray}
    |\mathcal{J}_N-\mathcal{J}_{N}^{\mathrm{id}}|\leq f(N)\sqrt{2\varepsilon_N}
\end{eqnarray}
where $\mathcal{J}_{N}^{\mathrm{id}}$ is the value of the functional \eqref{BE2} when states and the measurements $A_{i,0},A_{i,1}$ are exactly self-tested. Thus, we have that 
\begin{eqnarray}
    \mathcal{J}_N\leq\mathcal{J}_{N}^{\mathrm{id}}+ f(N)\sqrt{2\varepsilon_N}\leq \frac{1}{N-1}+ f(N)\sqrt{2\varepsilon_N}.
\end{eqnarray}
This completes the proof.

\end{proof}

\section{Appendix C: Numerics}

In this section we perform a numerical study of the real-complex gap of the scenario presented in the main text (see Figure 1). In order to be able to use tools of convex optimization, we generalize the star network to allow for shared classical randomness among the initial states. A gap in this more general scenario is also a gap in the original one, and thus we are able to considerably improve the analytical estimates in the case $N=3$.

Let us first rewrite the functionals of interest in a more convenient notation for the case $N=3$. For that, we call the non-central parties $A,B,C$ instead of $A_1,A_2,A_3$ respectively. Then, we have that Eqns. (3) and (6) of the main text become, respectively:
\[\hat{\mathcal{I}}'_l=
\hat{\mathcal{I}}_l\otimes E_l= (-1)^{l_1} (2A_1B_1C_1 + (-1)^{l_2} A_0B_0 + (-1)^{l_3} A_0C_0)\otimes E_l
\]
and
\[\hat{\mathcal{J}}'_3=
\hat{\mathcal{J}}_3\otimes E_0= -\frac{1}{3}(A_1B_2C_2+A_2B_2C_1+A_2B_1C_2)\otimes E_0,
\]
where $\{E_l\}_{l=0}^7$ are the central party's (Eve's) projectors which make up her measurement, and we use the usual convention in which tensor products between the different parties are ommitted. Also note that $A_1$ and $A_0$ are what elsewhere are called $\tilde{A}_{1,1}$ and $\tilde{A}_{1,0}$,respectively.

As mentioned in the main text, it is usually more convenient (in numerics and experiments) to work with a full version of $\mathcal{J}$ that accounts for all possible outcomes of Eve's measurement. This is, explicitly,
\[
\hat{\mathcal{J}}_3^{\text{full}} := -\frac{1}{3}\sum_{l=0}^7 ((-1)^{l_1+l_2+l_3}A_1B_2C_2+(-1)^{l_2}A_2B_2C_1 + (-1)^{l_3}A_2B_1C_2)E_l
\]

Using the notation $\mathcal{I}_l' :=\tr(\hat{\mathcal{I}}'_l \rho)$ and $\mathcal{J}^{\text{full}}_3:= \tr(\hat{\mathcal{J}}_3^{\text{full}} \rho) $, the optimization problems we consider in the main text are
\begin{align*}
\max&{\mathcal{J}^{\text{full}}_3}\\
\text{s.t. }& \mathcal{I}_l' = 1/2 \\
&\rho \text{ is in the star network}\\
&A_0,A_1,A_2,B_0,B_1,B_2,C_0,C_1,C_2 \text{ are binary observables} \\
&\{E_l\} \text{ is a projective measurement}\\
&\text{ All Hilbert spaces are arbitrary and real/complex}
\end{align*}

If the Hilbert spaces are complex, this has an optimal value of $1$ (with the operators mentioned in the main text). Since this is the algebraic maximum of $\mathcal{J}^{\text{full}}_3$, this is also true without the restriction of $\mathcal{I}_l = 1/2$.

In order to optimize in real Hilbert spaces, we use a variation of the NPA hierarchy presented in \cite{Marco}, which in turn was inspired by \cite{moroder2013device}. It works as follows. First, consider the formal real vector spaces $\mathcal{A}_n,\mathcal{B}_n,\mathcal{C}_n$ of words in the letters $\{A_0,A_1,A_2\}$, $\{B_0,B_1,B_2\}$ and $\{C_1,C_2,C_3\}$, respectively, with length $\leq n$. The map
\begin{eqnarray}
\Omega_A^n:\mathcal{T}(\mathcal{H}) &\rightarrow& \mathcal{B}(\mathcal{A}_n)  \\
\eta & \mapsto & \sum_{\alpha,\alpha'} \tr(\alpha\eta\alpha') \ket{\alpha}\!\bra{\alpha'}
\end{eqnarray}
is completely positive (respectively for $B$ and $C$). Therefore, for any state $\rho$ in the $ABC$ tripartite system the matrix

\[
\Omega_A^n \otimes \Omega_B^n \otimes \Omega_C^n (\rho_{ABC})
\]
is positive semi-definite. Note also that some of the entries of this matrix are precisely the summands in both our objective function and our restriction. Therefore, one can forget that one started with a quantum state and observables in the star network and just try to optimize the objective function with the restrictions that these matrices are positive semidefinite. This is an SDP optimization problem.

Now, consider an arbitrary state $\rho_{ABC}^E$ shared by systems $A,B,C,E$ in the star network (that is,
\[
\rho_{ABC}^E = \int  \rho^{E_1}_A(\lambda) \otimes \rho^{E_2}_B(\lambda) \otimes \rho^{E_3}_C(\lambda) d\lambda,
\]
for certain probability measure $\lambda$ which sum to $1$), and let $\rho^l_{ABC}:=\tr_E(E_l\rho^E_{ABC})$ the unnormalized state shared by systems $ABC$ after Eve's measurement, resulting in outcome $l$. The matrices
\[
\Gamma_n^\ell := \Omega_A^n \otimes \Omega_B^n \otimes \Omega_C^n (\rho_{ABC}^l)
\]
satisfy also that
\[
\sum_{l=0}^7 \Gamma_n^\ell = \Omega_A^n \otimes \Omega_B^n \otimes \Omega_C^n \left(\sum_{l=0}^7 \rho_{ABC}^l\right) = \Omega_A^n \otimes \Omega_B^n \otimes \Omega_C^n (\tr_E \rho_{ABC}^E)= \int \Omega_A^n(\rho_A(\lambda)) \otimes \Omega_B^n(\rho_B(\lambda)) \otimes \Omega_C^n(\rho_C(\lambda)) d\lambda,
\]
which is manifestly separable in the partition $A|B|C$. In usual quantum mechanics, this imposes complicated restrictions on $\Gamma_n^\ell$, but in $\mathbb{R}$QT, one obtains easily implementable linear restrictions from the fact that partial transpositions do not change the state:
\[
\left(\sum_{l=0}^7 \Gamma_n^\ell\right)^{T_{\mathcal{A}_n}} = \sum_{l=0}^7 \Gamma_n^\ell \qquad \left(\sum_{l=0}^7 \Gamma_n^\ell\right)^{T_{\mathcal{B}_n}} = \sum_{l=0}^7 \Gamma_n^\ell.
\]

Solving this SDP at level $2$, one obtains the following:

\begin{prop}\label{prop:numerics}
Consider the scenario depicted in Figure 1 of the main text with $N=3$. The real value of $\mathcal{J}^{\text{full}}_3$, conditioned on $\mathcal{I}_l' \geq 0.5-\varepsilon$ for all $l\in\{0,\ldots,7\}$ is upper bounded by $f(\varepsilon)$, as given in the following table:
\begin{table}[H]
\centering
\begin{tabular}{|c|c|c|c|c|c|c|}
\hline
$\varepsilon$ & $0$ & $0.05$ & $0.10$ & $0.15$ & $0.20$ &  $0.25$ \\
\hline
$f(\varepsilon)$ & $0.333333$ & $0.571594$& $0.761914$& $0.896554$ & $0.977448$ & $0.999999$ \\
\hline
\end{tabular}
\end{table}

(The unconditional complex value of $\mathcal{J}_3^{\text{full}}$ in the same scenario is always $1$, since this is the algebraic maximum and achieved with the example given in the main text).


\end{prop}

Note that this is manifestly better than the analytical results suggest. To compare this with Theorem $2$ please note that $\mathcal{I}_l$ is normalized in the main text by the probability of Eve obtaining outcome $l$, but $\mathcal{I}_l'$ is not. That is, $\varepsilon$ in this result is $\varepsilon_3/8$ of the main text. In Proposition \ref{prop:numerics}, also, we do not impose any restriction on the measurement of Eve other than it has $8$ outcomes whose probabilities of success add up to $1$.

We work with a variation of the NPA hierarchy already used in \cite{Marco}, adapted for this network. The code is written in Mathematica version 14.3, and runs on a Windows 10 Education version 22H2 with 16GB of RAM. For more details about the code itself, we refer the reader to the repository \cite{code}. The values in Proposition \ref{prop:numerics} have been obtained at level $2$ of the hierarchy with a duality gap of $10^{-6}$ or better after performing some standard symmetry reduction \cite{rosset2018symdpoly}.

\end{document}